\documentclass[a4paper,10pt]{article}

\usepackage{graphicx, algorithm, algorithmic}

\usepackage[]{amsmath}

\usepackage[]{amssymb}
\usepackage[]{amsthm}
\usepackage[]{pstricks}
\usepackage{gastex}
\usepackage[]{array} 
\usepackage[]{tabularx}

\usepackage{marvosym}

\usepackage{url}

\usepackage[low-sup]{subdepth}
\usepackage{enumitem}

\usepackage{stmaryrd}

\usepackage[]{hyperref}
\hypersetup{ 
colorlinks=true, 
linkcolor=red,
anchorcolor= black,
citecolor=green, 
filecolor=cyan,
menucolor=red,
runcolor=filecolor,
urlcolor=magenta}

\newtheorem{theorem}{Theorem}[section]
\newtheorem{proposition}[theorem]{Proposition}
\newtheorem{lemma}[theorem]{Lemma}
\newtheorem{corollary}[theorem]{Corollary}

{\theoremstyle{definition}
\newtheorem{example}{Example}[section]}

\newcommand{\bbN}{{\mathbb N}}     % Non negative integers
\renewcommand{\geq}{\geqslant}
\renewcommand{\leq}{\leqslant}

\renewcommand{\le}{\leqslant}

\newcommand{\Ima}{\operatorname{Im}}

\newcommand{\da}{\downarrow\!}

\newcommand{\hA}{\widehat{A^*}}
\newcommand{\hB}{\widehat{B^*}}
\newcommand{\hN}{\widehat{\bbN}}

\newcommand{\hX}{\widehat{X^*}}
\newcommand{\hf}{\widehat{f}}
\newcommand{\hphi}{\widehat{\varphi}}
\newcommand{\hpsi}{\widehat{\psi}}

\newcommand{\oL}{\overline{L}}

\newcommand{\cC}{\mathcal{C}}
\newcommand{\cL}{\mathcal{L}}

\newcommand{\cPd}{\mathcal{P}^{\downarrow}}
\newcommand{\cPdz}{\mathcal{P}_0^{\downarrow}}

\newcommand{\cV}{\mathcal{V}}
\newcommand{\cW}{\mathcal{W}}

\newcommand{\bSlp}{\mathbf{Sl}^\da}
\newcommand{\LcV}{\Lambda\mathcal{V}}
\newcommand{\LscV}{\Lambda'\mathcal{V}}
\newcommand{\bV}{\mathbf{V}}
\newcommand{\bW}{\mathbf{W}}

\newcommand{\PV}{\mathbf{PV}}
\newcommand{\PzV}{\mathbf{P}_0\mathbf{V}}
\newcommand{\Pd}{\mathbf{P}^{\downarrow}}
\newcommand{\Pdz}{\mathbf{P}_0^{\downarrow}}
\newcommand{\PdV}{\mathbf{P}^{\downarrow}\mathbf{V}}
\newcommand{\PdzV}{\mathbf{P}_0^{\downarrow}\mathbf{V}}

\newcommand{\cro}{\llbracket}
\newcommand{\crf}{\rrbracket}

\newcommand{\shu}{\mathop{\llcorner\!\llcorner\!\!\!\lrcorner}}

\newenvironment{conditions}
{%
	\begin{list}{\rm (\theenumi)}%
	{\noindent%
		\usecounter{enumi}%
		\setlength{\topsep}{2pt}%
		\setlength{\partopsep}{0pt}%
															 \setlength{\itemsep}{2pt}%
		\setlength{\parsep}{0pt}%
		\setlength{\leftmargin}{2.5em}%
		\setlength{\labelwidth}{1.5em}%
		\setlength{\labelsep}{0.5em}%
		\setlength{\listparindent}{0pt}%
		\setlength{\itemindent}{0pt}%
	}%
}%
{\end{list}}%

\newenvironment{conditionsprime}
{%
  \begin{list}{\rm (\theenumi$'$)}%
  {\noindent%
    \usecounter{enumi}%
    \setlength{\topsep}{2pt}%
    \setlength{\partopsep}{0pt}%
                \setlength{\itemsep}{2pt}%
    \setlength{\parsep}{0pt}%
    \setlength{\leftmargin}{2.5em}%
    \setlength{\labelwidth}{1.5em}%
    \setlength{\labelsep}{0.5em}%
    \setlength{\listparindent}{0pt}%
    \setlength{\itemindent}{0pt}%
  }%
}%
{\end{list}}%

\begin{document}

\title{Commutative positive varieties of languages \thanks{The last two authors
acknowledge support from the cooperation programme CNRS/Magyar Tudomanyos
Akad\'emia. The first author was partially supported by CMUP
(UID/MAT/00144/2013), which is funded by FCT (Portugal) with national (MCTES) and
European structural funds (FEDER), under the partnership agreement PT2020. The third
author was partially funded from the European Research Council (ERC) under the
European Union's Horizon 2020 research and innovation programme (grant agreement No
670624) and by the DeLTA project (ANR-16-CE40-0007)}}

% \author{Jorge Almeida\thanks{CMUP, Dep. Matem\'atica, Faculdade de Ci\^encias,
% Universidade do Porto, Rua do Campo Alegre 687, 4169-007 Porto, Portugal.
% \email{jalmeida@fc.up.pt}}, Zolt\'an \'Esik\thanks{Dept. of Computer Science,
% University of Szeged, \'Arp\'ad t\'er 2, H-6720 Szeged, P.O.B. 652 Hungary.}
% \and Jean-\'Eric Pin\thanks{IRIF, CNRS and Universit\'e Paris-Diderot, Case 7014,
% 75205 Paris Cedex 13, France.\email{Jean-Eric.Pin@irif.fr}}}

\author{\renewcommand\thefootnote{\arabic{footnote}} Jorge Almeida\footnotemark[1],
Zolt\'an \'Esik\footnotemark[2]$\ $ and Jean-\'Eric Pin\footnotemark[3]}

\date{\emph{To the memory of Zolt\'an \'Esik.}}

\footnotetext[1]{CMUP, Dep. Matem\'atica, Faculdade de Ci\^encias,
Universidade do Porto, Rua do Campo Alegre 687, 4169-007 Porto, Portugal.
\Email~\texttt{jalmeida@fc.up.pt}}

% \email{jalmeida@fc.up.pt}

\footnotetext[2]{Dept. of Computer Science, University of Szeged, \'Arp\'ad t\'er 2,
H-6720 Szeged, P.O.B. 652 Hungary.}

\footnotetext[3]{IRIF, CNRS and Universit\'e Paris-Diderot, Case 7014,
75205 Paris Cedex 13, France. \Email~\texttt{Jean-Eric.Pin@irif.fr}}

\maketitle

\begin{abstract}
	We study the commutative positive varieties of languages closed under various
	operations: shuffle, renaming and product over one-letter alphabets.
\end{abstract}

\vspace{10mm}

Most monoids considered in this paper are finite. In particular, we use the term
\emph{variety of monoids} for \emph{variety of finite monoids}. Similarly, all
languages considered in this paper are regular languages and hence their syntactic
monoid is finite.

\medskip

%%%%%%%%%%%%%%%%%%%%
%                  %
%   Introduction   %
%                  %
%%%%%%%%%%%%%%%%%%%%

\section{Introduction}\label{sec:Introduction}

Eilenberg's variety theorem \cite{Eilenberg76} and its ordered version \cite{Pin95}
provide a convenient setting for studying classes of regular languages. It states
that positive varieties of languages are in one-to-one correspondence with varieties
of finite ordered monoids. 

There is a large literature on operations on regular languages. For instance, the
closure of [positive] varieties of languages under various operations has been
extensively studied: Kleene star \cite{Perrot78}, concatenation product
\cite{BrancoPin09, PinStraubing05, Straubing79b}, renaming \cite{Almeida95,
AlmeidaCanoKlimaPin15, CanoPin12, Reutenauer79, Straubing79} and shuffle
\cite{BerstelBoassonCartonPinRestivo10, CanoPin04, EsikSimon98}. The ultimate goal
would be the complete classification of the positive varieties of languages closed
under these operations. The first step in this direction is to understand the
commutative case, which is the goal of this paper.

We first show in Theorem \ref{thm:ld varieties are varieties} that every commutative
positive $ld$-variety of languages is a positive variety of languages. This means
that if a class of commutative languages is closed under Boolean operations and under
inverses of length-decreasing morphisms then it is also closed under inverses of
morphisms. This result has a curious application in weak arithmetic, stated in
Proposition \ref{prop:arithmetic result}. 

Next we study two operations on languages, shuffle and renaming. These two operations
are closely related to the so-called \emph{power operator} on monoids, which
associates with each monoid the monoid of its subsets. In its ordered version, it
associates with each ordered monoid the ordered monoid of its downsets. We give four
equivalent conditions characterizing the commutative positive varieties of languages
closed under shuffle (Proposition \ref{prop:shuffle}) or under renaming (Proposition
\ref{prop:renaming}).

In order to keep the paper self-contained, prerequisites are presented in some detail
in Section \ref{sec:Prerequisites}. Inequalities form the topic of Section
\ref{sec:Inequalities and identities}. We start with their formal definitions,
describe their various interpretations and establish some of their properties.
General results on renaming are given in Section \ref{sec:Renaming} and
more specific results on commutative varieties are proposed in Section
\ref{sec:Commutative varieties}, including our previously mentioned result on
$ld$-varieties. Our characterizations of the positive varieties of languages closed
under shuffle or renaming form the meat of Section \ref{sec:Operations on commutative
languages} and are illustrated by three examples in Section \ref{sec:Three examples}.
Finally, a few research directions are suggested in Section \ref{sec:Conclusion}.

%%%%%%%%%%%%%%%%%%%%%
%                   %
%   Prerequisites   %
%                   %
%%%%%%%%%%%%%%%%%%%%%

\section{Prerequisites}\label{sec:Prerequisites}

In this section, we briefly recall the following notions: lattices and (positive)
varieties of languages, syntactic ordered monoids, varieties of ordered monoids, 
stamps, downset monoids, free profinite monoids.

%%%%%%%%%%%%%%%%%
%               %
%   Languages   %
%               %
%%%%%%%%%%%%%%%%%

\subsection{Languages}\label{subsec:Languages}

Let $A$ be a finite alphabet. Let $[u]$ be the \emph{commutative closure} of a word
$u$, that is, the set of words commutatively equivalent to $u$. For instance, $[aab]
= \{aab, aba, baa\}$. A language $L$ is \emph{commutative} if, for every word $u \in
L$, $[u]$ is contained in $L$.

A \emph{lattice of languages} is a set $\cL$ of regular languages of $A^*$ containing
$\emptyset$ and $A^*$ and closed under finite union and finite intersection. It is
\emph{closed under quotients} if, for each $L \in \cL$ and $u \in A^*$, the languages
$u^{-1}L$ and $Lu^{-1}$ are also in $\cL$.

The \emph{shuffle product} (or simply \emph{shuffle}) of two languages $L_1$ and
$L_2$ over $A$ is the language
\begin{multline*}
	L_1 \shu L_2 = \{ w \in A^* \mid w = u_1v_1 {}\dotsm{} u_nv_n \text{ for
	some words $u_1, \ldots, u_n$} \\
	\text{$v_1, \ldots, v_n$ of $A^*$ such that $u_1 {}\dotsm{} u_n \in
	L_1$ and $v_1 {}\dotsm{} v_n \in L_2$} \}
\end{multline*}
The shuffle product defines a commutative and associative operation on the set of
languages over $A$.

A \emph{renaming} or \emph{length-preserving morphism} is a morphism $\varphi$ from
$A^*$ into $B^*$, such that, for each word $u$, the words $u$ and $\varphi(u)$ have
the same length. It is equivalent to require that, for each letter $a$, $\varphi(a)$
is also a letter, that is, $\varphi(A) \subseteq B$. Similarly, a morphism is
\emph{length-decreasing} if the image of each letter is either a letter or the empty
word.

A \emph{class of languages} is a correspondence $\cC$ which associates with
each alphabet $A$ a set $\cC(A^*)$ of regular languages of $A^*$. 

A \emph{positive variety of languages} is a class of regular languages $\cV$ such
that:
\begin{conditions}
	\item \label{item:lattice} for every alphabet $A$, $\cV(A^*)$ is a lattice of
	languages closed under quotients,

	\item \label{item:inverse morphism} if $\varphi: A^* \rightarrow B^*$ is a
	morphism, $L\in \cV(B^*)$ implies $\varphi^{-1}(L)\in {\cal V}(A^*)$.
\end{conditions}
A \emph{variety of languages} is a positive variety $\cV$ such that each lattice
$\cV(A^*)$ is closed under complement. We shall also use two slight variations of
these notions. A \emph{positive $ld$-variety \textup{[}$lp$-variety\textup{]} of languages}
% \emph{positive $ld$-variety $\ob lp$-variety$\cb$ of languages}
\cite{EsikIto03, PinStraubing05} is a class of regular languages $\cV$ satisfying
\eqref{item:lattice} and
\begin{conditionsprime}
	\setcounter{enumi}{1} 	
	\item \label{item:inverse ld-morphism} if $\varphi: A^* \rightarrow B^*$ is a
	length-decreasing [length-preserving] morphism, then $L\in \cV(B^*)$ implies
	$\varphi^{-1}(L)\in \cV(A^*)$.
\end{conditionsprime}

%%%%%%%%%%%%%%%%%%%%%%%%%%%%%%%%%
%                               %
%   Syntactic ordered monoids   %
%                               %
%%%%%%%%%%%%%%%%%%%%%%%%%%%%%%%%%

\subsection{Syntactic ordered monoids}\label{subsec:Syntactic ordered monoids}

An \emph{ordered monoid} is a monoid $M$ equipped with a partial order $\leq$
compatible with the product on $M$: for all $x, y, z \in M$, if $x \leq y$ then $zx
\leq zy$ and $xz \leq yz$.

The \emph{ordered syntactic monoid} of a language was first introduced by M.P.
Sch\"utzenberger in \cite[p.~10]{Schutzenberger56}. Let $L$ be a language of $A^*$.
The \emph{syntactic preorder} of $L$ is the relation $\leq_L$ defined on $A^*$ by
$u \leq_L v$ if, for every $x, y \in A^*$, $xuy \in L$ implies $xvy \in 
L$. When the language $L$ is clear from the context, we may write $\le$
instead of~$\le_L$. As is standard in preorder notation, we write $u <
v$ to mean that $u \le v$ holds but $v\le u$ does not.

For instance, let $A = \{a\}$. If $L = a + a^3$, then $a^3 \leq_L a$,
but if $L = a + a^3a^*$, then $a \leq_L a^3$.

The associated equivalence relation $\sim_L$, defined by $u \sim_L v$ if $u \leq_L v$
and $v \leq_L u$, is the \emph{syntactic congruence} of $L$ and the quotient monoid
$M(L) = A^*/{\sim_L}$ is the \emph{syntactic monoid} of $L$. The natural morphism
$\eta: A^* \to A^*/{\sim_L}$ is the \emph{syntactic stamp} of $L$. The
\emph{syntactic image} of $L$ is the set $P = \eta(L)$.

The \emph{syntactic order} $\leq$ is defined on $M(L)$ as follows: $u \leq v$ if and
only if for all $x,y \in M$, $xuy \in P$ implies $xvy \in P$. The partial order
$\leq$ is compatible with multiplication and the resulting ordered monoid $(M, \leq)$
is called the \emph{ordered syntactic monoid} of $L$.
\begin{example}\label{ex:1 + a}
Let $L$ be the language $1 + a$. The syntactic monoid of $L$ is the commutative
monoid $\{1, a, 0\}$ satisfying $a^2 = 0$. The syntactic order is $0 < a < 1$.
Indeed, one has $a \le 1$ since, for each $r \geq 0$, the condition $a^ra \in L$
implies $a^r \in L$. Similarly, one has $0 \le a$ since, for each $r \geq 0$, the
condition $a^r a^2 \in L$ implies $a^ra \in L$. However, $1 \not\leq a$ and $a
\not\leq 0$ since $a \in L$ but $a^2 \notin L$.
\end{example}
\begin{example}\label{ex:a + a6a*}
Let $L$ be the language $a + a^6a^*$. The syntactic monoid of $L$ may be identified
with the commutative monoid $\{0, 1, \ldots, 6\}$ equipped with the operation $xy =
\min\{x +y, 6\}$. In particular, $0$ and $6$ are the unique idempotents. 
The syntactic order is represented as follows (a path from $i$ to
$j$ means that $i < j$):
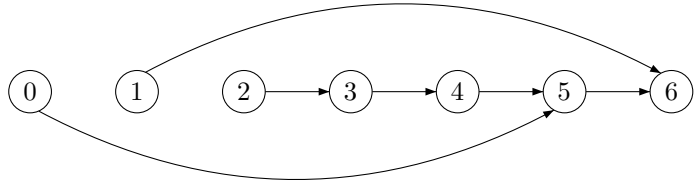
\begin{center}
	\thinlines
	\unitlength=4pt
	\begin{picture}(60,22)(0,-11)
		\gasset{Nw=4,Nh=4,Nmr=2,curvedepth=0}
		\node(A0)(0,0){$0$}
		\node(A1)(10,0){$1$}
		\node(A2)(20,0){$2$}
		\node(A3)(30,0){$3$}
		\node(A4)(40,0){$4$}
		\node(A5)(50,0){$5$}
		\node(A6)(60,0){$6$}
		\drawedge[curvedepth=7,syo=1.4,eyo=1.2](A1,A6){} 
		\drawedge(A5,A6){}
		\drawedge(A4,A5){}
		\drawedge(A3,A4){}
		\drawedge(A2,A3){}
		\drawedge[curvedepth=-7,syo=-1.4,eyo=-1.2](A0,A5){}
	\end{picture}
\end{center}
For instance, one has $1 < 6$ since, for each $r \geq 0$, the condition $a a^r \in L$
implies $a^6 a^r \in L$. Similarly, one has $0 < 5$ since, for each $r \geq 0$, the
condition $a^r \in L$ implies $a^5 a^r \in L$. But $1 \not< 5$ since $a \in L$ but
$a^5 \notin L$.
\end{example}
\begin{example}\label{ex:a + (a3 + a4)(a7)*}
Let $L$ be the language $a + (a^3 + a^4)(a^7)^*$. Its minimal automaton is 
represented below.
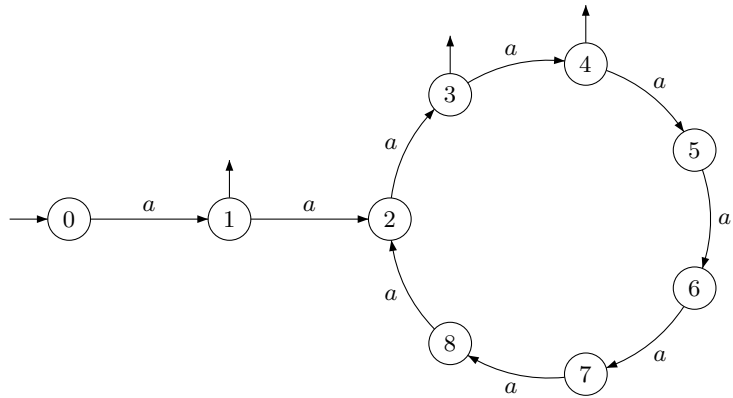
\begin{center}
	\thinlines
	\unitlength=4pt
	\begin{picture}(58.514533,38)(0,-19)
		\gasset{Nw=4,Nh=4,Nmr=2,curvedepth=0}\small
		\node[Nmarks=i,iangle=180](A1)(0,0){$0$}
		\node[Nmarks=f,fangle=90](A2)(15,0){$1$}
%%%%%%%%%%%%%%%%%%%%%%%%%%%%%%%%%%%%%%%%%%%%%%%%%%%%%%%%%%%%%%%%%%%%%%%%%%%%%%%%%%%%%%%%%%%%%
%                                                                                           %
%   Coordonnées de Ai: (40 + 15 * cos((13-2i)\pi/7), 15 * sin((13-2i)\pi/7))                %
%   En Sage                                                                                 %
%     for n in range(7):                                                                    %
%       print "\\node(A" + str(n+3) + ")(" + str(45 + 15 * cos(pi*(13-2*(n+3))/7).n(30))    %
%       + "," + str(15 * sin(pi*(13-2*(n+3))/7).n(30)) + "){$" + str(n+3) + "$}"            %
%                                                                                           %
%%%%%%%%%%%%%%%%%%%%%%%%%%%%%%%%%%%%%%%%%%%%%%%%%%%%%%%%%%%%%%%%%%%%%%%%%%%%%%%%%%%%%%%%%%%%%		
		\node(A3)(30.000000,0.00000000){$2$}
		\node[Nmarks=f,fangle=90](A4)(35.647653,11.727472){$3$}
		\node[Nmarks=f,fangle=90](A5)(48.337814,14.623919){$4$}
		\node(A6)(58.514533,6.5082561){$5$}
		\node(A7)(58.514533,-6.5082561){$6$}
		\node(A8)(48.337814,-14.623919){$7$}
		\node(A9)(35.647653,-11.727472){$8$}
		\drawedge(A1,A2){$a$}
		\drawedge(A2,A3){$a$}
		\drawedge[curvedepth=1.5](A3,A4){$a$}
		\drawedge[curvedepth=1.5](A4,A5){$a$}
		\drawedge[curvedepth=1.5](A5,A6){$a$}
		\drawedge[curvedepth=1.5](A6,A7){$a$}
		\drawedge[curvedepth=1.5](A7,A8){$a$}
		\drawedge[curvedepth=1.5](A8,A9){$a$}
		\drawedge[curvedepth=1.5](A9,A3){$a$}
	\end{picture}
\end{center}
The syntactic monoid of $L$ is the monoid presented by $\langle a \mid a^9 = a^2
\rangle$. The syntatic order is the equality relation.
\end{example}

%%%%%%%%%%%%%%
%            %
%   Stamps   %
%            %
%%%%%%%%%%%%%%

\subsection{Stamps}\label{subsec:Stamps}

Monoids and ordered monoids are used to recognise languages, but there is a slightly
more restricted notion. A \emph{stamp} is a surjective monoid morphism $\varphi:A^*
\to M$ from a finitely generated free monoid $A^*$ onto a finite monoid $M$.
If $M$ is an ordered monoid, $\varphi$ is called an \emph{ordered stamp}.

The \emph{restricted direct product} of two [ordered] stamps $\varphi_1:A^* \to
M_1$ and $\varphi_2:A^* \to M_2$ is the stamp $\varphi$ with domain
$A^*$ defined by $\varphi(a) = (\varphi_1(a), \varphi_2(a))$ (see
Figure \ref{fig:restricted product}). The image of $\varphi$ is an
[ordered] submonoid of the [ordered] monoid $M_1 \times M_2$.
\begin{figure}[H]
	\begin{center}
		\thinlines
		\unitlength=4pt
		\begin{picture}(35,24)(0,2)
			\gasset{Nframe=n,Nw=6,Nh=5,Nmr=2.5,curvedepth=0}\small
			\node(A)(0,13){$A^*$}
			\node(M1)(35,0){$M_1$}
			\node(M2)(35,26){$M_2$}
			\node[Nw=23](M12)(25,13){$\Ima(\varphi) \subseteq M_1 \times M_2$}
			\drawedge[ELside=r](A,M1){$\varphi_1$}
			\drawedge(A,M2){$\varphi_2$}
			\drawedge[sxo=-7](A,M12){$\varphi$}
			\drawedge[sxo=-4](M12,M1){$\pi_1$}
			\drawedge[sxo=-4,ELside=r](M12,M2){$\pi_2$}
		\end{picture}
	\end{center}
	\caption{The restricted direct product of two stamps.}\label{fig:restricted
	product}
\end{figure}
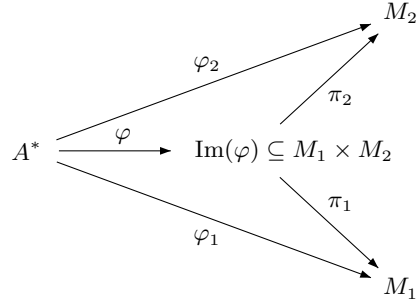
\noindent Recall that an \emph{upset} of an ordered set $E$ is a subset $U$ of $E$
such that the conditions $x \in U$ and $x \leq y$ imply $y \in U$. A language $L$ of
$A^*$ is \emph{recognised by a stamp $\varphi: A^* \to M$} if there exists a subset
$P$ of $M$ such that $L = \varphi^{-1}(P)$. It is \emph{recognised by an ordered
stamp $\varphi: A^* \to M$} if there exists an upset $U$ of $M$ such that $L =
\varphi^{-1}(U)$.

It is easy to see that if two languages $L_0$ and $L_1$ of $A^*$ are recognised by 
the [ordered] stamps $\varphi_0$ and $\varphi_1$, respectively, then $L_0 \cap L_1$ 
and $L_0 \cup L_1$ are both recognised by the restricted product of $\varphi_0$ and 
$\varphi_1$.

%%%%%%%%%%%%%%%%%
%               %
%   Varieties   %
%               %
%%%%%%%%%%%%%%%%%

\subsection{Varieties}\label{subsec:Varieties}

Varieties of languages and their avatars all admit an algebraic characterization. We
first describe the corresponding algebraic objects and summarize the correspondence
results at the end of this section. See \cite{Pin12} for more details.

[Positive] varieties of languages correspond to varieties of [ordered] monoids. A
\emph{variety of monoids} is a class of monoids closed under taking submonoids,
quotients and finite direct products. \emph{Varieties of ordered monoids} are defined
analogously.

The description of the algebraic objects corresponding to positive $lp$- and
$ld$-varieties of languages is more complex and relies on the notion of stamp defined
in Section \ref{subsec:Stamps}. An \emph{$lp$-morphism} from a stamp $\varphi : A^*
\rightarrow M$ to a stamp $\psi: B^* \rightarrow N$ is a pair $(f, \alpha)$, where
$f: A^*\rightarrow B^*$ is length-preserving, $\alpha: M \rightarrow N$ is a morphism
of [ordered] monoids, and $\psi \circ f = \alpha \circ
\varphi$.
	\begin{center}
		\thinlines
		\unitlength=4pt
		\begin{picture}(12,16)(0,-1)
			\gasset{Nframe=n,Nw=5,Nh=5,Nmr=2.5,curvedepth=0}\small
			\node(A)(0,12){$A^*$}
			\node(B)(12,12){$B^*$}
			\node(M)(0,0){$M$}
			\node(N)(12,0){$N$}
			\drawedge(A,B){$f$}
			\drawedge[ELside=r](A,M){$\varphi$}
			\drawedge(B,N){$\psi$}
			\drawedge(M,N){$\alpha$}
		\end{picture}
	\end{center}
The $lp$-morphism $(f, \alpha)$ is an \emph{$lp$-projection} if $f$ is surjective. It
is an \emph{$lp$-inclusion} if $\alpha$ is injective. 

An \emph{\textup{[}ordered\textup{]} $lp$-variety of stamps} is a class of [ordered]
stamps closed under $lp$-projections, $lp$-inclusions and finite restricted direct
products. \emph{$\textup{[}$Ordered\textup{]} $ld$-varieties of stamps} are defined
in the same way, just by replacing $lp$ by $ld$ and length-preserving by
length-decreasing everywhere in the definition.

Here are the announced correspondence results. Eilenberg's variety theorem
\cite{Eilenberg76} and its ordered counterpart \cite{Pin95} give a bijective
correspondence between varieties of [ordered] monoids and positive varieties of
languages. Let $\bV$ be a variety of finite [ordered] monoids and, for each alphabet
$A$, let $\cV(A^*)$ be the set of all languages of $A^*$ whose [ordered] syntactic
monoid is in $\bV$. Then $\cV$ is a [positive] variety of languages. Furthermore, the
correspondence $\bV \rightarrow \cV$ is a bijection between varieties of [ordered]
monoids and [positive] varieties of languages.

There is a similar correspondence for $lp$-varieties of [ordered] stamps
\cite{EsikIto03, Straubing02}. Let $\bV$ be an $lp$-variety of [ordered] stamps. For
each alphabet $A$, let $\cV(A^*)$ be the set of all languages of $A^*$ whose
[ordered] syntactic stamp is in $\bV$. Then $\cV$ is a [positive] $lp$-variety of
languages. Furthermore, the correspondence $\bV \rightarrow \cV$ is a bijection
between $lp$-varieties of [ordered] stamps and [positive] $lp$-varieties of
languages.

Finally, there is a similar statement for $ld$-varieties of [ordered] stamps.

%%%%%%%%%%%%%%%%%%%%%%%
%                     %
%   Downset monoids   %
%                     %
%%%%%%%%%%%%%%%%%%%%%%%

\subsection{Downset monoids}\label{subsec:Downset monoids}

Let $(M,\leq)$ be an ordered monoid. A \emph{downset} of $M$ is a subset $F$ of
$M$ such that if $x \in F$ and $y \leq x$ then $y \in F$. The \emph{product of two
downsets} $X$ and $Y$ is the downset
$$
	XY = \{ z \in M \mid \text{there exist $x \in X$ and $y \in Y$ such that $z \leq
	xy$} \}
$$
This operation makes the set of nonempty downsets of $M$ a monoid, denoted by
$\cPd(M)$ and called the \emph{downset monoid} of $M$. Its identity element is $\da
1$. If one also considers the empty set, one gets a monoid with zero, denoted
$\cPdz(M)$, in which the empty set is the zero. For instance, if $M$ is the trivial
monoid, $\cPdz(M)$ is isomorphic to the ordered monoid $\{0, 1\}$, consisting of an
identity $1$ and a zero $0$, ordered by $0 < 1$. This monoid will be denoted by
$U_1^\da$.

The monoids $\cPdz(M)$ and $\cPd(M)$ are closely related. First, $\cPd(M)$ is a
submonoid of $\cPdz(M)$. Secondly, as shown in \cite[Proposition~5.1,
p.~452]{CanoPin04}, $\cPdz(M)$ is isomorphic to a quotient monoid of $\cPd(M) \times
U_1^\da$.

The monoids $\cPd(M)$ and $\cPdz(M)$ are naturally ordered by inclusion, denoted by
$\leq$. Note that $X \leq Y$ if and only if, for each $x \in X$, there exists $y \in
Y$ such that $x \leq y$.

Given a variety of ordered monoids $\bV$, let $\PdV$ [$\PdzV$] denote the variety of
ordered monoids generated by the monoids of the form $\cPd(M)$ [$\cPdz(M)$], where $M
\in \bV$. The operator $\Pd$ was intensively studied in \cite{AlmeidaCanoKlimaPin15}.
In particular, it is known that both $\Pd$ and $\Pdz$ are idempotent operators.

The hereinabove relation between $\cPdz(M)$ and $\cPd(M)$ can be extended to
varieties as follows. Let $\bSlp$ be the variety of ordered monoids generated by
$U_1^\da$. It is a well-known fact that $\bSlp = \cro xy = yx, x = x^2, x \leq
1\crf$. Moreover, the equality
\begin{equation}
	\PdzV = \PdV \vee \bSlp	\label{eq:PdzV}
\end{equation}
holds for any variety of ordered monoids $\bV$.

%%%%%%%%%%%%%%%%%%%%%%%%%%%%%
%                           %
%   Free profinite monoid   %
%                           %
%%%%%%%%%%%%%%%%%%%%%%%%%%%%%

\subsection{Free profinite monoid}\label{subsec:Free profinite monoid}

We refer the reader to \cite{Almeida95, Almeida02, Almeida05, Weil02} for detailed
information on profinite completions and we just recall here a few useful facts. Let
$d$ be the profinite metric on the free monoid $A^*$. We let $\hA$ denote the
completion of the metric space $(A^*, d)$. The product on $A^*$ is uniformly
continuous and hence has a unique continuous extension to $\hA$. It follows that
$\hA$ is a compact monoid, called the \emph{free profinite monoid} on $A$.
Furthermore, every stamp $\varphi: A^* \to M$ admits a unique continuous extension
$\hphi: \hA \to M$. Similarly, every morphism $f:A^* \to B^*$ admits a unique
continuous extension $\hf: \hA \to \hB$. In the sequel, $\oL$ denotes the closure in
$\hA$ of a subset $L$ of $A^*$.

The length of a word $u$ is denoted by $|u|$. The length map $u \to |u|$ defines a
morphism from $A^*$ to the additive semigroup $\bbN$. If $A = \{a\}$, this morphism
is actually an isomorphism, which maps $a^n$ to $n$. In other words, $(\bbN, +, 0)$
is the free monoid with a single generator. We let $\hN$ denote the profinite
completion of $\bbN$, which is of course isomorphic to $\widehat{a^*}$.

This allows one to define the length $|u|$ of an element $u$ of $\hA$ simply by 
extending by continuity the length map defined on $A^*$. The length map is actually 
a morphism, that is, $|1| = 0$ and $|uv| = |u| + |v|$ for all $u, v \in \hA$. 

%%%%%%%%%%%%%%%%%%%%%%%%%%%%%%%%%%%
%                                 %
%   Identities and inequalities   %
%                                 %
%%%%%%%%%%%%%%%%%%%%%%%%%%%%%%%%%%%

\section{Inequalities and identities}\label{sec:Inequalities and identities}

The inequalities [equalities] occurring in this paper are of the form $u \leq v$ $[u
= v]$, where $u$ and $v$ are both in $\hA$ for some alphabet $A$. In an ordered
context, $u = v$ is often viewed as a shortcut for $u \leq v$ and $v \leq u$.

However, these inequalities are interpreted in several different contexts, which may
confuse the reader. Let us clarify matters by giving precise definitions for each
case.

%%%%%%%%%%%%%%%%%%%%
%                  %
%   Inequalities   %
%                  %
%%%%%%%%%%%%%%%%%%%%

\subsection{Inequalities}\label{subsec:Inequalities}

\paragraph{Ordered monoids.}

Let $M$ be an ordered monoid, let $X$ be an alphabet and let $u, v \in \hX$. Then $M$
\emph{satisfies the inequality $u \leq v$} if, for each morphism $\psi: X^* \to M$,
$\hpsi(u) \leq \hpsi(v)$.

This is the formal definition but in practice, it is easier to think of $u$ and $v$
as terms in which one substitutes each symbol $x \in X$ for an element of $M$. For
instance, $M$ satisfies the inequality $xy^{\omega+1} \leq x^\omega y$ if, for all
$x, y \in M$, $xy^{\omega+1} \leq x^\omega y$.

\paragraph{Varieties of ordered monoids.}

Let $\bV$ be a variety of ordered monoids, let $X$ be an alphabet and let $u, v \in
\hX$. Then $\bV$ \emph{satisfies an inequality $u \leq v$} if each ordered monoid of
$\bV$ satisfies the inequality. In this context, equalities of the form $u = v$ are
often called \emph{identities}.

It is proved in \cite{PinWeil96b} that any variety of ordered monoids may be defined
by a (possibly infinite) set of such inequalities. This result extends to the ordered
case the classical result of Reiterman \cite{Reiterman82} and Banaschewski
\cite{Banaschewski83}: any variety of monoids may be defined by a (possibly infinite)
set of identities.

\paragraph{The case of $lp$-varieties and $ld$-varieties of ordered stamps.}

Let $\bV$ be an $lp$-variety [$ld$-variety] of ordered stamps, let $X$ be an alphabet
and let $u, v \in \hX$. Then $\bV$ \emph{satisfies the inequality $u \leq
v$} if, for each stamp $\varphi:A^* \to M$ of $\bV$ and for every length-preserving
[length-decreasing] morphism $f: X^* \to A^*$, $\hphi(\hf(u)) \leq \hphi(\hf(v))$.

The difficulty is to interpret correctly $\hf(u)$. If $f$ is length-preserving,
$\hf(u)$ is obtained by replacing each symbol $x \in X$ by a letter of $A$. For
instance, an $lp$-variety $\bV$ satisfies the inequality $xy^{\omega+1} \leq x^\omega
y$ if, for each stamp $\varphi:A^* \to M$ of $\bV$ and for all letters $a, b \in A$,
$\hphi(ab^{\omega+1}) \leq \hphi(a^\omega b)$.

It is proved in \cite{Kunc03,PinStraubing05} that any ordered $lp$-variety of stamps
may be defined by a (possibly infinite) set of such inequalities.

If $f$ is length-decreasing, this is even more tricky. Then $\hf(u)$ is obtained by
replacing each symbol $x \in X$ by either a letter of $A$ or by the empty word. For
instance, an $ld$-variety $\bV$ satisfies the inequality $xy^{\omega+1} \leq x^\omega y$
if, for each stamp $\varphi:A^* \to M$ of $\bV$ and for all letters $a, b \in A$,
$\hphi(ab^{\omega+1}) \leq \hphi(a^\omega b)$, $\hphi(b^{\omega+1}) \leq \hphi(b)$ 
and $\hphi(a) \leq \hphi(a^\omega)$.

It is proved in \cite{Kunc03,PinStraubing05} that any ordered $ld$-variety of stamps
may be defined by a (possibly infinite) set of such inequalities.

\medskip

\noindent We will also need the following elementary result. Recall that a variety of
[ordered] monoids is \emph{aperiodic} if it satisfies the identity $x^\omega =
x^{\omega + 1}$. 

\begin{proposition}\label{prop:aperiodic}
	Let $\bV$ be an aperiodic variety of ordered monoids. Then, for each $\alpha \in
	\hN$, $\bV$ satisfies the identity $x^\omega = x^\omega x^\alpha$.
\end{proposition}

\begin{proof} Let $\alpha \in \hN$. Then $\alpha = \lim_{n \to \infty} k_n$ for some
sequence $(k_n)_{n \geq 0}$ of nonegative integers. Since $\bV$ is aperiodic, it
satisfies the identity $x^{\omega + k_n} = x^\omega$ for all $n$, and hence it also
satisfies the identity $x^\omega x^\alpha = x^\omega$.\end{proof}

%%%%%%%%%%%%%%%%%%%%%%%%%%%%
%                          %
%   Shuffle and renaming   %
%                          %
%%%%%%%%%%%%%%%%%%%%%%%%%%%%

\section{Renaming}\label{sec:Renaming}

In this section, we give some general results on renaming. 

Since any map may be written as the composition of an injective map with a surjective
map, one gets immediately:

\begin{lemma}\label{lem:renaming = injective and surjective}
	A class of languages is closed under renaming if and only if it is closed under
	injective and surjective renamings.
\end{lemma} 

The next two results give a simple description of the positive $lp$-varieties
[$ld$-varieties] of languages closed under injective renaming:

\begin{proposition}\label{prop:lp-varieties closed under renaming}
The following conditions are equivalent for a positive $lp$-variety of languages
$\cV$:
\begin{conditions}
  \item \label{item:injective renaming} $\cV$ is closed under injective renaming, 
	
	\item \label{item:nonempty B*} for each alphabet $A$ and each nonempty set
	$B \subseteq A$, $B^*$ belongs to $\cV(A^*)$,

	\item \label{item:B*} for each alphabet $A$ and each set $B
	\subseteq A$, $B^*$ belongs to $\cV(A^*)$.
\end{conditions}  
\end{proposition}

\begin{proof} \eqref{item:injective renaming} implies \eqref{item:B*}. Suppose that
$\cV$ is closed under injective renaming. Let $B$ be a subset of an alphabet $A$.
Since $B^* \in \cV(B^*)$ and since the embedding of $B^*$ into $A^*$ is an injective
renaming, one also has $B^* \in \cV(A^*)$.

\eqref{item:B*} implies \eqref{item:nonempty B*} is trivial.

\eqref{item:nonempty B*} implies \eqref{item:B*}. We have to show that for any
alphabet $A$, $\{1\} \in \cV(A^*)$. First assume that $A$ has at least two elements.
If $A = B_1 \cup B_2$ is a partition of $A$ into two disjoint nonempty sets $B_1$ and
$B_2$, then both $B_1^*$ and $B_2^*$ are in $\cV(A^*)$, so that $\{1\} = B_1^* \cap
B_2^*$ is also in $\cV(A^*)$. Now consider a one-letter alphabet $a$ and the
two-letter alphabet $\{a, b\}$. The inclusion $h: a^* \to \{a,b\}^*$ is length
preserving and thus $\{1\} = h^{-1}(\{1\})$ is in $\cV(a^*)$. Finally, the result is
trivial if $A$ is empty.

\eqref{item:B*} implies \eqref{item:injective renaming}. Suppose that, for each
alphabet $A$ and nonempty set $B \subseteq A$, $B^* \in \cV(A^*)$. Let $h: B^* \to
A^*$ be an injective renaming. Then there is a renaming $f: A^* \to B^*$ such that $f
\circ h$ is the identity function on $B^*$. Since for any $L \subseteq B^*$, $h(L) =
{f^{-1}(L) \cap (h(B))^*}$, we conclude that $h(L) \in \cV(A^*)$ whenever $L \in
\cV(B^*)$.\end{proof}

\begin{proposition}\label{prop:ld-varieties closed under renaming}
	An $ld$-variety $\cV$ is closed under injective renaming if and only if for each
	one-letter alphabet $a$, $\{1\}$ belongs to $\cV(a^*)$.
\end{proposition} 

\begin{proof} Since each $ld$-variety is an $lp$-variety, Proposition
\ref{prop:lp-varieties closed under renaming} shows that $\cV$ is closed under
injective renaming if and only if, for each alphabet $A$ and each subset $B$ of $A$,
$B^*$ belongs to $\cV(A^*)$. In particular, if $\cV$ is closed under injective
renaming, then $\{1\}$ belongs to $\cV(a^*)$.

Suppose now that $\cV(a^*)$ contains $\{1\}$. Let $A$ be any alphabet and let $B$ be
a subset of $A$. The morphism $h:A^* \to a^*$ that maps each element of $B$ to $1$
and all elements of $A \setminus B$ to $a$ is length-decreasing. Since $\cV$ is an
$ld$-variety and $\{1\}$ belongs to $\cV(a^*)$, $h^{-1}(\{1\})$ also belongs to
$\cV(a^*)$. But $B^* = h^{-1}(\{1\})$, and hence $\cV(A^*)$ contains $B^*$ as
required.\end{proof}

Let $\bV$ be a variety of ordered monoids and let $\cV$ be the corresponding positive
variety of languages. A description of the positive variety of languages
corresponding to $\PdV$ was given by Pol\'{a}k \cite[Theorem 4.2]{Polak02} and by
Cano and Pin \cite{Cano03} and \cite[Proposition 6.3]{CanoPin04}. The following
stronger version\footnote{We warn the reader that a different notation was used in
\cite{CanoPin12}.} was given in \cite{CanoPin12}. For each alphabet $A$, let us denote
by $\LcV(A^*)$ [$\LscV(A^*)$] the set of all languages of $A^*$ of the form
$\varphi(K)$, where $\varphi$ is a [surjective] renaming from $B^*$ to
$A^*$, $B$~is an arbitrary finite alphabet, and $K$
is a language of $\cV(B^*)$.

\begin{theorem}\label{thm:renaming}
	The class $\LcV$ $[\LscV]$ is a positive variety of languages and the corresponding
	variety of ordered monoids is $\PdzV$ $[\PdV]$.
\end{theorem}

\begin{corollary}\label{cor:renaming}
	A positive variety of languages $\cV$ is closed under
	\textup{[}surjective\textup{]} renaming if and only if $\bV = \PdzV$ $[\bV =
	\PdV]$.
\end{corollary}

%%%%%%%%%%%%%%%%%%%%%%%%%%%%%
%                           %
%   Commutative varieties   %
%                           %
%%%%%%%%%%%%%%%%%%%%%%%%%%%%%

\section{Commutative varieties}\label{sec:Commutative varieties}

A stamp $\varphi: A^* \to M$ is said to be \emph{commutative} if $M$ is commutative.
An $ld$-variety is \emph{commutative} if all its stamps are
commutative. A stamp $\varphi: A^* \to M$ is called \emph{monogenic}
if $A$ is a singleton alphabet.

\begin{proposition}\label{prop:commutative ld-varieties}  
	Every commutative $ld$-variety of \textup{[}ordered\textup{]} stamps is generated
	by its monogenic \textup{[}ordered\textup{]} stamps.
\end{proposition}

\begin{proof} We first give the proof in the unordered case. Let $\bV$ be a
commutative $ld$-variety of stamps and let $\varphi:A^* \to M$ be a stamp of $\bV$.
For each $a \in A$, denote by $M_a$ the submonoid of $M$ generated by $\varphi(a)$
and let $\gamma_a: A^* \to M_a$ be the stamp defined by $\gamma_a(a) = \varphi(a)$
and $\gamma_a(c) = 1$ for $c \neq a$. Let $\bW$ be the $ld$-variety of stamps
generated by the stamps $\gamma_a$, for $a \in A$. We claim that $\bV = \bW$.

Let $\pi_a: A^* \to A^*$ be the length-decreasing morphism defined by $\pi_a(a) = a$
and $\pi_a(c) = 1$ for $c \neq a$. Denoting by $\iota_a$ the natural embedding from
$M_a$ into $M$, one gets the
following commutative diagram:
\begin{center}
	\thinlines
	\unitlength=4pt
	\begin{picture}(16,17)(0,-2)
		\gasset{Nframe=n,Nw=6,Nh=6,Nmr=3,curvedepth=0}\small
		\node(A)(0,12){$A^*$}
		\node(B)(16,12){$A^*$}
		\node(Ma)(0,0){$M_a$}
		\node(M)(16,0){$M$}
		\drawedge(A,B){$\pi_a$}
		\drawedge[ELside=r](A,Ma){$\gamma_a$}
		\drawedge(B,M){$\varphi$}
		\drawedge(Ma,M){$\iota_a$}
	\end{picture}
\end{center}
Therefore $(\pi_a, \iota_a)$ is an $ld$-inclusion and each stamp
$\gamma_a$ belongs to $\bV$. Thus $\bW \subseteq \bV$. 

The restricted product $\gamma$ of the stamps $\gamma_a$ also belongs to $\bW$. Note
that $\gamma$ is a surjective morphism from $A^*$ onto $\prod_{a\in A}M_a$. Moreover,
the function $\alpha:\prod_{a \in A}M_a \to M$ which maps each family $(m_a)_{a \in
A}$ onto the product $\prod_{a \in A}m_a$ is a surjective morphism. Since $\alpha
\circ \gamma = \varphi$, the stamp $\varphi$ belongs to $\bW$. Thus $\bV \subseteq
\bW$. This proves the claim and the proposition.

In the ordered case, each $M_a$ is an ordered submonoid of $M$ and thus each
$\gamma_a$ is an ordered stamp. Since $\iota_a$ clearly preserves the order, the same
argument shows that each $\gamma_a$ is in $\bV$ and thus $\bW \subseteq \bV$. For the
reverse inclusion, one basically needs to observe that $\prod_{a \in A}M_a$ is
equipped with the product order, and that the map $\alpha$ preserves the order, since
$M$ is an ordered monoid.\end{proof}

A similar but simpler proof would give the following result:

\begin{proposition}\label{prop:commutative varieties}  
	Every commutative variety of \textup{[}ordered\textup{]} monoids is generated by its
	monogenic \textup{[}ordered\textup{]} monoids.
\end{proposition}

Proposition \ref{prop:commutative ld-varieties} has an interesting consequence in
terms of languages. Equivalently, a language is \textit{commutative} if its syntactic
monoid is commutative.

\begin{corollary}\label{cor:commutative ld-varieties}
	Let $\cV_1$ and $\cV_2$ be two positive $ld$-varieties of commutative languages.
	Then $\cV_1 \subseteq \cV_2$ if and only if $\cV_1(a^*) \subseteq \cV_2(a^*)$.
\end{corollary}

Corollary \ref{cor:commutative ld-varieties} shows that a positive commutative
$ld$-variety of languages is entirely determined by its languages on a one-letter
alphabet. Here is a more explicit version of this result.

\begin{proposition}\label{prop:cV}
	Let $\cV$ be a commutative positive $ld$-variety of languages. Then for each
	alphabet $A = \{a_1,\ldots,a_k\}$, $\cV(A^*)$ consists of all finite unions of
	languages of the form $L_1 \shu {}\dotsm{} \shu L_k$ where, for $1 \leq i \leq k$,
	$L_i \in \cV(a_i^*)$.
\end{proposition}

\begin{proof} Let $A = \{a_1,\ldots,a_k\}$ be an alphabet. Let $\cW(A^*)$
consist of all finite unions of languages of the form $L_1 \shu {}\dotsm{}
\shu L_k$ where, for $1 \leq i \leq k$, $L_i \in \cV(a_i^*)$. Let us 
first prove a lemma.

\begin{lemma}\label{lem:W}
	The class $\cW$ is a commutative positive $ld$-variety of languages.
\end{lemma}

\begin{proof} By construction, every language of $\cW$ is commutative. Furthermore,
$\cW(A^*)$ is closed under union. To prove that $\cW(A^*)$ is closed under
intersection, it suffices to show that the intersection of any two languages $L =
L_1 \shu {}\dotsm{} \shu L_k$ and $L' = L'_1 \shu {}\dotsm{} \shu L'_k$ with $L_i,
L_i' \in \cV(a_i^*)$ is in $\cW(A^*)$. We claim that
\begin{equation}\label{eq:intersection shuffle}
  L \cap L' = {(L_1 \cap L_1')} \shu {}\dotsm{} \shu {(L_k \cap L_k')}
\end{equation}
Let $R$ be the right hand side of \eqref{eq:intersection shuffle}. The inclusion $R
\subseteq L \cap L'$ is clear. Moreover, if $u \in L \cap L'$, then $u \in (a_1^{n_1}
\shu {}\dotsm{} \shu a_k^{n_k}) \cap (a_1^{n'_1} \shu {}\dotsm{} \shu a_k^{n'_k})$, with
$a_i^{n_i} \in L_i$ and $a_i^{n'_i} \in L'_i$ for $1 \leq i \leq k$. This forces $n_i =
n'_i$ and hence $u \in R$, which proves the claim.

Let us prove that $\cW(A^*)$ is closed under quotient by any word $u$. Setting $n_i =
|u|_{a_i}$ for $1 \leq i \leq k$, it suffices to observe that
$$
	u^{-1}(L_1 \shu {}\dotsm{} \shu L_k) = (a_1^{n_1})^{-1}L_1 \shu
	{}\dotsm{} \shu {(a_k^{n_k})^{-1}L_k}
$$
Finally, let $\alpha:B^* \to A^*$ be a length-decreasing morphism.
It is proved in \cite[Proposition 
1.1]{BerstelBoassonCartonPinRestivo10} that
\begin{equation}\label{eq:inverse morphism shuffle}
	\alpha^{-1}(L_1 \shu {}\dotsm{} \shu L_k) = \alpha^{-1}(L_1) \shu
	{}\dotsm{} \shu \alpha^{-1}(L_k)
\end{equation}
It follows that $\cW$ is closed under inverses of $ld$-morphisms, 
which concludes the proof.\end{proof}

\noindent Let us now come back to the proof of Proposition \ref{prop:cV}. Since $\cW$
is a commutative positive $ld$-variety by Lemma \ref{lem:W}, it suffices to prove, by
Proposition \ref{prop:commutative ld-varieties}, that $\cV(a^*) = \cW(a^*)$ for each
one-letter alphabet $a$. But this follows from the definition of $\cW$.\end{proof}

Proposition \ref{prop:cV} has an interesting consequence.

\begin{theorem}\label{thm:ld varieties are varieties}
	Every commutative positive $ld$-variety of languages is a positive variety of
	languages.
\end{theorem}

\begin{proof} Let $\cV$ be a commutative positive $ld$-variety of languages and let
$\cW$ be the positive variety of languages generated by $\cV$. We claim that $\cV =
\cW$. Since $\cV$ is contained in $\cW$, Corollary \ref{cor:commutative ld-varieties}
shows that it suffices to prove that $\cW(a^*) \subseteq \cV(a^*)$ for each
one-letter alphabet $a$. Since inverses of morphisms commute with Boolean operations
and quotients, it suffices to prove that if $\varphi:a^* \to A^*$ is a morphism and
$L \in \cV(A^*)$, then $\varphi^{-1}(L) \in \cV(a^*)$. 

Let $\varphi(a) = a_1 {}\dotsm{} a_k$, where $a_1, \ldots, a_k$ are letters of the
alphabet $A$. Setting $C = \{c_1, \ldots, c_k\}$, where $c_1, \ldots, c_k$ are
distinct letters, one may write $\varphi$ as $\alpha \circ \beta$ where $\beta:a^*
\to C^*$ is defined by $\beta(a) = c_1 {}\dotsm{} c_k$ and $\alpha:C^* \to A^*$ is
defined by $\alpha(c_i) = a_i$ for $1 \leq i \leq k$.
\begin{center}
	\thinlines
	\unitlength=4pt
	\begin{picture}(20,10)(0,-7)
		\gasset{Nframe=n,Nw=4,Nh=4,Nmr=2,curvedepth=0}\small
		\node(Aa)(0,0){$a^*$}
		\node(C)(10,0){$C^*$}
		\node(A)(20,0){$A^*$}
		\drawedge[curvedepth=-5,ELside=r](Aa,A){$\varphi$} 
		\drawedge(Aa,C){$\beta$}
		\drawedge(C,A){$\alpha$}
	\end{picture}
\end{center}
Since $\alpha$ is length-preserving, the language $K = \alpha^{-1}(L)$ belongs to
$\cV(C^*)$. It follows by Proposition \ref{prop:cV} that $K$ is a finite union of
languages of the form $L_1 \shu {}\dotsm{} \shu L_k$ where, for $1 \leq i \leq k$,
$L_i \in \cV(c_i^*)$. Let, for $1 \leq i \leq k$, $\beta_i$ be the unique length
preserving morphism from $a^*$ to $c_i^*$, defined by $\beta_i(a^r) = c_i^r$. 
We claim that
\begin{equation}\label{eq:inverse beta}
	\beta^{-1}(L_1 \shu {}\dotsm{} \shu L_k) = \beta_1^{-1}(L_1) \cap {}\dotsm{}
	\cap \beta_k^{-1}(L_k)
\end{equation}
Let $R$ be the right hand side of \eqref{eq:inverse beta}. If $a^r \in R$, then
$\beta_i(a^r) \in L_i$. Therefore $c_i^r \in L_i$ and since $\beta(a^r) = (c_1
{}\dotsm{} c_k)^r$, $\beta(a^r) \in L_1 \shu {}\dotsm{} \shu L_k$. Thus $R$ is a
subset of $\beta^{-1}(L_1 \shu {}\dotsm{} \shu L_k)$.

If now $a^r \in \beta^{-1}(L_1 \shu {}\dotsm{} \shu L_k)$, then $\beta(a^r) \in
c_1^{n_1} \shu {}\dotsm{} \shu c_k^{n_k}$ with $c^{n_i} \in L_i$ for $1 \leq i \leq
k$. But since $\beta(a^r) = (c_1 {}\dotsm{} c_k)^r$, one has $n_1 = {}\dotsm{} = n_k
= r$ and hence $c_i^r \in L_i$. Therefore $a^r \in \beta_i^{-1}(L_i)$ for all $i$ and
thus $a^r$ belongs $R$. This proves \eqref{eq:inverse beta}.
 
Since $L_i \in \cV(c_i^*)$ and $\beta_i$ is length-preserving, $\beta_i^{-1}(L_i) \in
\cV(a^*)$. As $K$ is a finite union of languages of the form $L_1 \shu {}\dotsm{}
\shu L_k$, Formula \eqref{eq:inverse beta} shows that $\beta^{-1}(K) \in \cV(a^*)$.
Finally, since $\varphi = \alpha \circ \beta$, one gets $\varphi^{-1}(L) =
\beta^{-1}(\alpha^{-1}(L)) = \beta^{-1}(K)$. Therefore $\varphi^{-1}(L) \in
\cV(a^*)$, which concludes the proof.\end{proof}

Theorem \ref{thm:ld varieties are varieties} has a curious interpretation on the set
of natural numbers, mentioned in \cite{CegielskiGrigorieffGuessarian14}. Setting,
for each subset $L$ of $\bbN$ and each positive integer $k$,
\begin{align*}
	L-1 &= \{n \in \bbN \mid n + 1 \in L \} \\
	L \div k &= \{n \in \bbN \mid kn \in L \}
\end{align*}
one gets the following result:

\begin{proposition}\label{prop:arithmetic result}
	Let $\cL$ be a lattice of finite subsets\footnote{It also works for a lattice of
	regular subsets of $\bbN$.} of $\bbN$ such that if $L \in \cL$, then $L-1 \in
	\cL$. Then for each positive integer $k$, $L \in \cL$ implies $L \div k \in \cL$.
\end{proposition}

%%%%%%%%%%%%%%%%%%%%%%%%%%%%%%%%%%%%%%%%%%%
%                                         %
%   Operations on commutative languages   %
%                                         %
%%%%%%%%%%%%%%%%%%%%%%%%%%%%%%%%%%%%%%%%%%%

\section{Operations on commutative languages}\label{sec:Operations on
commutative languages}

In this section, we compare the expressive power of three operations on commutative
languages: product, shuffle and renaming. 

%%%%%%%%%%%%%%%
%             %
%   Shuffle   %
%             %
%%%%%%%%%%%%%%%

\subsection{Shuffle}\label{subsec:Shuffle}

Let us say that a positive variety of languages $\cV$ is \emph{closed under product
over one-letter alphabets} if, for each one-letter alphabet $a$, $\cV(a^*)$ is closed
under product. Commutative positive varieties closed under shuffle may be described
in various ways.

\begin{proposition}\label{prop:shuffle}
	Let $\cV$ be a commutative positive variety of languages and let $\bV$ be the
	corresponding variety of ordered monoids. The following conditions are equivalent:
\begin{conditions}
	\item \label{item:surjective renaming} $\cV$ is closed under surjective renaming,

	\item \label{item:shuffle product} $\cV$ is closed under shuffle product,

	\item \label{item:product} $\cV$ is closed under product over one-letter alphabets,
	
	\item \label{item:V = PuV} $\bV = \PdV$.
\end{conditions}
\end{proposition}
	
\begin{proof} 

\eqref{item:surjective renaming} implies \eqref{item:shuffle product}. Let $B = A
\times \{0, 1\}$ and let $\pi_0$, $\pi_1$ and $\pi$ be the three morphisms from $B^*$
to $A^*$ defined for all $a \in A$ by
\begin{alignat*}{3}
	\pi_0(a, 0) &= a &\qquad \pi_1(a, 0) &= 1 &\qquad \pi(a, 0) &= a \\
	\pi_0(a, 1) &= 1 &       \pi_1(a, 1) &= a &       \pi(a, 1) &= a 
\end{alignat*}
Let $L_0$ and $L_1$ be two languages of $A^*$. Since $\pi$ is a surjective renaming,
the formula $L_0 \shu L_1 = \pi(\pi^{-1}_0(L_0) \cap \pi^{-1}_1(L_1))$ shows that
every positive variety closed under surjective renaming is closed under shuffle
product.

\eqref{item:shuffle product} implies \eqref{item:product} is trivial since on a
one-letter alphabet, shuffle product and product are the same.

\eqref{item:product} implies \eqref{item:surjective renaming}. Let $\pi: A^* \to B^*$
be a surjective renaming. For each $b \in B$, let $\gamma_b:b^* \to a^*$ be the
renaming which maps $b$ onto $a$. Let $L$ be a language of $\cV(A^*)$. By Proposition
\ref{prop:cV}, $L$ is a finite union of languages of the form $\shu_{a \in A} L_a$
where $L_a \in \cV(a^*)$ for each $a \in A$. For each $b \in B$, let
$$
  K_b = \prod_{a \in \pi^{-1}(b)}\gamma_b^{-1}(L_a)
$$
If $\cV(a^*)$ is closed under product for each one-letter alphabet $a$, then $K_b$
belongs to $\cV(b^*)$. Finally, the formula $\pi(L) = \shu_{b \in B} K_b$ shows that
$\pi(L)$ belongs to $\cV(B^*)$. Therefore $\cV$ is closed under surjective renaming.

Finally, the equivalence of \eqref{item:surjective renaming} and \eqref{item:V = 
PuV} follows from Corollary \ref{cor:renaming}.\end{proof}

%%%%%%%%%%%%%%%%
%              %
%   Renaming   %
%              %
%%%%%%%%%%%%%%%%

\subsection{Renaming}\label{subsec:Renaming}

\noindent Let us say that a positive variety of languages contains $\{1\}$ if, for
every alphabet $A$, $\cV(A^*)$ contains the language $\{1\}$. The following result 
is a slight variation on Proposition \ref{prop:shuffle}.

\begin{proposition}\label{prop:renaming}
	Let $\cV$ be a commutative positive variety of languages and let $\bV$ be the
	corresponding variety of ordered monoids. The following conditions are equivalent:
\begin{conditions}
	\item \label{item:renaming} $\cV$ is closed under renaming,

	\item \label{item:surjective renaming and 1} $\cV$ is closed under surjective
	renaming and contains $\{1\}$,

	\item \label{item:shuffle product and 1} $\cV$ is closed under shuffle product and
	contains $\{1\}$,

	\item \label{item:product and 1} $\cV$ is closed under product over one-letter
	alphabets and contains $\{1\}$,

	\item \label{item:V = PzV} $\bV = \PdzV$.
\end{conditions}
\end{proposition}

\begin{proof} The equivalence of \eqref{item:surjective renaming and
1}---\eqref{item:product and 1} follows directly from Proposition \ref{prop:shuffle}.
If \eqref{item:surjective renaming and 1} holds, then $\cV$ is closed under injective
renaming by Proposition \ref{prop:ld-varieties closed under renaming} and hence is
closed under renaming by Lemma \ref{lem:renaming = injective and surjective}. Thus
\eqref{item:surjective renaming and 1} implies \eqref{item:renaming}.

To show that \eqref{item:renaming} implies \eqref{item:surjective renaming and 1}, it
suffices to show that if $\cV$ is closed under renaming then it contains $\{1\}$. Let
$A = \{a, b\}$ and let $\pi: A^* \to A^*$ be the renaming defined by $\pi(a) = \pi(b)
= a$. Since $A^* \in \cV(A^*)$ and $\pi(A^*) = a^*$, one has $a^* \in \cV(A^*)$. A
similar argument would show that $b^* \in \cV(A^*)$ and thus the language $\{1\}$,
which is the intersection of $a^*$ and $b^*$ also belongs to $\cV(A^*)$. Consider now
an alphabet $B$ and the morphism $\alpha$ from $B^*$ to $A^*$ defined by $\alpha(c) =
a$ for each $c \in B$. Then $\alpha^{-1}(\{1\}) = \{1\}$ and thus $\cV$ contains
$\{1\}$.

Finally, the equivalence of \eqref{item:renaming} and \eqref{item:V = PzV} follows
from Corollary \ref{cor:renaming}.\end{proof}

%%%%%%%%%%%%%%%%%%%%%%
%                    %
%   Three examples   %
%                    %
%%%%%%%%%%%%%%%%%%%%%%

\section{Three examples}\label{sec:Three examples}

In this section, we study the positive varieties of languages generated by the
languages of Examples \ref{ex:1 + a}, \ref{ex:a + a6a*} and \ref{ex:a + (a3 +
a4)(a7)*}.

%%%%%%%%%%%%%%%%%%%%%%%%%%%%
%                          %
%   The language $1 + a$   %
%                          %
%%%%%%%%%%%%%%%%%%%%%%%%%%%%

\subsection{The language \texorpdfstring{$1 + a$}{1 + a}}\label{subsec:1 + a}

Let $L$ be the language $1 + a$, let $M$ be its ordered syntactic monoid and let
$\cV$ be the smallest commutative positive variety such that $\cV(a^*)$ contains $L$.
Let $\bV$ be the variety of finite ordered monoids corresponding to $\cV$.

Since a positive variety of languages is closed under quotients, $\cV(a^*)$ contains
the language $a^{-1}L = 1$. It follows that $\cV(a^*)$ contains 4 languages:
$\emptyset$, $1$, $1 + a$ and $a^*$. We claim that
\[
   \bV = \cro\, \text{$xy = yx$, $x \leq 1$ and $x^2 \leq x^3$}\, \crf.
\]
First, the two inequalities $x \leq 1$ and $x^2 \leq x^3$ hold in $M$. Furthermore,
the inequality $x \leq 1$ implies the inequalities of the form $x^p \leq x^q$ with $p
> q$ and the inequality $x^2 \leq x^3$ implies all the inequalities of the form $x^p
\leq x^q$ with $2 \leq p < q$. The only other nontrivial inequalities that $\bV$
could possibly satisfy are $1 \leq x^q$ for $q > 0$ or $x \leq x^q$ for $q > 1$.
However, $M$ does not satisfy any of these inequalities.

Let $\cV'$ be the closure of $\cV$ under shuffle, or equivalently, under product over
one-letter alphabets. Then $\cV'(a^*)$ contains the empty language, the language
$a^*$ and all languages of the form $(1 + a)^n$ with $n \geq 0$. By Theorem
\ref{thm:renaming} and Proposition \ref{prop:shuffle}, $\cV'$ corresponds to the
variety of ordered monoids $\PdV$. We claim that
\[
	\PdV = \cro\, \text{$xy = yx$ and $x \leq 1$}\, \crf.
\]
Indeed, the ordered syntactic monoids of the languages of $\cV'(a^*)$ all satisfy $xy
= yx$ and $x \leq 1$. Conversely, if the ordered syntactic monoid of a language $K$
of $a^*$ satisfies $x \leq 1$, then $x^n \leq_K 1$ for every $n \geq 0$, and
$K$ is closed under taking subwords. If $K$ is infinite, this forces $K = a^*$.
If $K$ is finite, it is necessarily of the form $(1 + a)^n$ with $n \geq 0$. In both
cases, $K$ belongs to $\cV'(a^*)$.

Finally, let $\bW$ be the variety of ordered monoids corresponding to the closure of
$\cV$ under renaming. Since $U_1^\da \in \PdV$, Theorem \ref{thm:renaming} and
Formula \eqref{eq:PdzV} show that
\[
	\bW = \PdzV = {\PdV \vee \bSlp} = \PdV = \cro\, \text{$xy = yx$ and $x \leq 1$}\,
	\crf.
\]

%%%%%%%%%%%%%%%%%%%%%%%%%%%%%%%%%
%                               %
%   The language $a + a^6a^*$   %
%                               %
%%%%%%%%%%%%%%%%%%%%%%%%%%%%%%%%%

\subsection{The language \texorpdfstring{$a + a^6a^*$}{a + a6a*}}\label{subsec:a + a6a*}

Let $L$ be the language $a + a^6a^*$, let $M$ be its ordered syntactic monoid and let
$\cV$ be the smallest commutative positive variety such that $\cV(a^*)$ contains $L$.
Let $\bV$ be the variety of finite ordered monoids corresponding to $\cV$.

Since a positive variety of languages is closed under quotients, $\cV(a^*)$ contains
the language $a^{-1}L = 1 + a^5a^*$ and the language ${L \cap a^{-1}L} = a^6a^*$. It
also contains the quotients of this language, which are the languages $a^ja^*$, for
$j \leq 6$. Taking the union with $L$, $a^{-1}L$ or both, one finally concludes that
$\cV(a^*)$ contains 20 languages: $\emptyset$, $a^ia^*$ for $0 \leq i \leq 6$, $1 +
a^ia^*$ for $1 \leq i \leq 5$, $a + a^ia^*$ for $3 \leq i \leq 6$ and $1 + a +
a^ia^*$ for $ 3 \leq i \leq 5$.

We claim that
\[
   \bV = \cro xy = yx, 1 \leq x^5, x^2 \leq x^3, x^6 = x^7 \crf.
\]
Indeed, all defining inequalities hold in $M$. Since $x^6 = x^7$, the other possible
inequalities satisfied by $M$ are equivalent to an inequality of the form $x^p \leq
x^q$ with $p < q \leq 6$. For $p = 0$, the only inequalities of this form satisfied
by $M$ are $1 \leq x^5$ and $1 \leq x^6$, but $1 \leq x^6$ is a consequence of $1
\leq x^5$ and $x^2 \leq x^3$ since $1 \leq x^5 = x^3x^2 \leq x^3x^3 = x^6$. For $p =
1$, the only inequality of this form satisfied by $M$ is $x \leq x^6$, which is a
consequence of $1 \leq x^5$. Finally, the inequality $x^2 \leq x^3$ implies $x^p \leq
x^q$ for $2 \leq p < q \leq 6$.

Let $\cV'$ be the closure of $\cV$ under shuffle, or equivalently, under product over
one-letter alphabets. We claim that $\cV'(a^*)$ consists of the empty set and the
languages of the form
\begin{equation}
  a^n(F + a^5a^*) \label{eq:languages}
\end{equation}
where $n \geq 0$ and $F$ is a subset of $(1 + a)^4$. First of all, the languages of
the form \eqref{eq:languages} and the empty set form a lattice closed under product,
since if $0 \leq n \leq m$ and $F$ and $G$ are subsets of $(1 + a)^4$, then
\begin{align*}
  a^n(F + a^5a^*) + a^m(G + a^5a^*) &= a^n(F + a^{m-n}G + a^5a^*) \\
	a^n(F + a^5a^*) \cap a^m(G + a^5a^*) &= a^m\Bigl(\bigl((a^{m-n})^{-1}(F
	+ a^5a^*)\bigr) \cap G\bigr) + a^5a^*\biggr)\\
  a^n(F + a^5a^*) a^m(G + a^5a^*) &= a^{n+m}(FG + a^5a^*)
\end{align*}
Since $\cV'(a^*)$ is closed under finite unions, it just remains to prove that the
languages of the form $a^n(a^k + a^5a^*)$, with $ n \geq 0$ and $0 \leq k \leq 4$ all
belong to $\cV'(a^*)$. But since the languages $a + a^6a^*$ and $1 + a^{5-k}a^*$ are
in $\cV(a^*)$, this follows from the formula
\[
  a^n(a^k + a^5a^*) = \bigl(a + a^6a^*)^{n+k}(1 + a^{5-k}a^*)
\]
By Theorem \ref{thm:renaming} and Proposition \ref{prop:shuffle}, $\cV'$ corresponds
to the variety of ordered monoids $\PdV$. We claim that
\[
	\PdV = \cro\, \text{$xy = yx$ and $1 \leq x^n$ for $5 \leq n \leq 9$}\, \crf.
\]
Indeed, the ordered syntactic monoid of any of the languages of the form
\eqref{eq:languages} satisfies all inequalities of the form $1 \leq x^n$ for $n \geq
5$, but the syntactic ordered monoid of $1 + a^2a^*$ does not satisfy any inequality
of the form $x^p \leq x^q$ with $p > q$. Moreover, the only inequalities that are not
an immediate consequence of an inequality of the form $1 \leq x^n$ with $5 \leq n
\leq 9$ are the inequalities $x^i \leq x^j$ with $0 \leq j - i \leq 4$. But none of these
inequalities are satisfied by the ordered syntactic monoid of $a^i(1 + a^5a^*)$.

Finally, Theorem \ref{thm:renaming} and Formula \eqref{eq:PdzV} show that the variety
of ordered monoids corresponding to the closure of $\cV$ under renaming is
\begin{align*}
	\PdzV &= {\PdV \vee \bSlp} \\
	&= \cro\, \text{$xy = yx$ and $1 \leq x^n$ for $5 \leq n \leq 9$}\, \crf \vee \cro xy
	= yx, x^2 = x, x \leq 1 \,\crf.
\end{align*}
We claim that $\PdzV= \bW$, where
\[
	\bW = \cro\, \text{$xy = yx$ and $x \leq x^n$ for $6 \leq n \leq 10$}\, \crf.
\]
First, the inequality $x \leq x^n$ is a consequence both of the inequality $1 \leq
x^{n-1}$ and of the equation $x= x^2$. It follows that $\PdzV \subseteq \bW$. To
establish the opposite inclusion, it suffices to establish the claim that any
inequality of the form $x^p \leq x^q$ satisfied by both $\PdV$ and $\bSlp$ is also
satisfied by $\bW$. If $p = 0$, then the inequality becomes $1 \leq x^q$ and it is
not satisfied by $\bSlp$ since $1 \not< 0$ in $U_1^\da$. Moreover, for $p > 0$, the
only inequalities of the form $x^p \leq x^q$ that are not an immediate consequence of an
inequality of the form $x \leq x^n$ with $6 \leq n \leq 10$ are the inequalities $x^p
\leq x^q$ with $0 \leq q - p \leq 4$. But we already observed that the ordered
syntactic monoid of $a^p(1 + a^5a^*)$ belongs to $\PdV$ but does not satisfy any of
these inequalities, which proves the claim.

%%%%%%%%%%%%%%%%%%%%%%%%%%%%%%%%%%%%%%%%%%%%%
%                                           %
%   The language $a + (a^3 + a^4)(a^7)^*$   %
%                                           %
%%%%%%%%%%%%%%%%%%%%%%%%%%%%%%%%%%%%%%%%%%%%%

\subsection{The language \texorpdfstring{$a + (a^3 + a^4)(a^7)^*$}{a + (a3 +
a4)(a7)*}}\label{subsec:a + (a3 + a4)(a7)*}

Let $L$ be the language $a + (a^3 + a^4)(a^7)^*$, let $M$ be its ordered syntactic
monoid and let $\cV$ be smallest commutative positive variety such that $\cV(a^*)$
contains $L$. Let $\bV$ be the variety of finite ordered monoids corresponding to
$\cV$. One has
\begin{align*} 
	(a)^{-1}L &= 1 + (a^2 + a^3)(a^7)^* & (a^2)^{-1}L &= (a + a^2)(a^7)^* \\
  (a^3)^{-1}L &= (1 + a)(a^7)^* & (a^4)^{-1}L &= (1+a^6)(a^7)^* \\
  (a^5)^{-1}L &= (a^5 + a^6)(a^7)^* & (a^6)^{-1}L &= (a^4 + a^5)(a^7)^* \\
  (a^7)^{-1}L &= (a^3 + a^4)(a^7)^* & (a^8)^{-1}L &= (a^2 + a^3)(a^7)^* 
\end{align*}
The set of final states of the minimal automaton of $L$ is $\{1, 3, 4\}$.
The quotients of $L$ are recognised by the same automaton by taking a different set 
of final states as indicated below
\begin{align*} 
	(a)^{-1}L &\to \{0, 2, 3\} & (a^2)^{-1}L &\to \{1, 2, 8\} \\
  (a^3)^{-1}L &\to \{0, 1, 7, 8\} & (a^4)^{-1}L &\to \{0, 6, 7\} \\
  (a^5)^{-1}L &\to \{5, 6\} & (a^6)^{-1}L &\to \{4, 5\} \\
  (a^7)^{-1}L &\to \{3, 4\} & (a^8)^{-1}L &\to \{2, 3\} 
\end{align*}
Observing that 
\begin{align*} 
	\{0\} &= \{0, 2, 3\} \cap \{0, 6, 7\} &
        \{1\} &= \{1, 3, 4\} \cap \{1, 2, 8\} \\
	\{2\} &= \{0, 2, 3\} \cap \{1, 2, 8\} &
	\{3\} &= \{1, 3, 4\} \cap \{0, 2, 3\} \\
        \{4\} &= \{3, 4\} \cap \{4, 5\} & 
	\{5\} &= \{4, 5\} \cap \{5, 6\} \\
	\{6\} &= \{5, 6\} \cap \{0, 6, 7\} &
        \{0, 7\} &= \{0, 6, 7\} \cap \{0, 1, 7, 8\} \\
	\{1, 8\} &= \{1, 2, 8\} \cap \{0, 1, 7, 8\}
\end{align*}
it follows that a language belongs to the lattice of languages generated by the 
quotients of $L$ if and only if it is accepted by the minimal automaton of $L$ 
equipped with a set $F$ of final states satisfying the two conditions 
\begin{equation}
  7 \in F \implies 0 \in F \quad \text{and} \quad 8 \in F \implies 1 \in F 
	\label{eq:F}
\end{equation}
Now, the complement of a set $F$ satisfying \eqref{eq:F} also satisfies \eqref{eq:F}.
It follows that the lattice of languages generated by the quotients of $L$ is
actually a Boolean algebra and consequently, $\cV$ is a variety of languages. It
also follows that
\[
   \bV = \cro xy = yx, x^2 = x^9 \crf.
\]
Moreover, since $U_1 = \{0, 1\}$ belongs to $\bV$, it follows that $\PV = \PzV$. By
\cite[Th\'eor\`eme 2.14]{Perrot78}, $\PV$ is the variety of all commutative monoids
whose groups satisfy the identity $x^7 = 1$. Therefore
\[
   \PV = \cro xy = yx, x^\omega = x^{\omega + 7} \crf.
\]
The closure of $\cV$ under shuffle, or equivalently, under product over one-letter
alphabets, and the closure of $\cV$ under renaming both correspond to the variety of
monoids $\PV$.

\section{Conclusion}\label{sec:Conclusion}

We gave an algebraic characterization of the commutative positive varieties of
languages closed under shuffle product, renaming or product over one-letter
alphabets, but several questions might be worth a further study.

First, each commutative variety of ordered monoids can be described by the equality
$xy = yx$ and by a set of inequalities in one variable, like $x^p \leq x^q$ or more
generally $x^\alpha \leq x^\beta$ with $\alpha, \beta \in \hN$. It would then be
interesting to compare these varieties. We just mention a few results of this
flavour, which may help in finding bases of inequalities for commutative positive
varieties of languages.

\begin{proposition}\label{prop:inequalities x xn} 
	The variety $\cro xy = yx, x \leq x^{n+1}\crf$ is contained in the variety $\cro xy
	= yx, x \leq x^{m+1}\crf$ if and only if $n$ divides $m$.
\end{proposition}

\begin{proof} Suppose that $n$ divides $m$, that is, $m = kn$ for some $k \geq 0$. If
$x \leq x^{n+1}$, then $x \leq xx^n$ and by induction, $x \leq xx^{kn} = xx^m =
x^{m+1}$. Thus $\cro xy = yx, x \leq x^{n+1}\crf$ is contained in the variety $\cro
xy = yx, x \leq x^{m+1}\crf$.
	
Suppose now that $\cro xy = yx, x \leq x^{n+1}\crf$ is contained in the variety $\cro
xy = yx, x \leq x^{m+1}\crf$. Then the ordered syntactic monoid of $a(a^n)^*$
satisfies the inequality $x \leq x^{n+1}$ and thus it also satisfies the inequality
$x \leq x^{m+1}$. Since $a \in a(a^n)^*$, this means in particular that $a^m \in
a(a^n)^*$ and thus that $n$ divides $m$.
\end{proof}

\noindent In fact, a more general result holds. For each set of natural numbers $S$,
let
\[
  \bV_S = \cro\, xy = yx, x \leq x^{n+1} \text{ for all $n \in S$}\, \crf.
\]
Let $\langle S \rangle$ denote the additive submonoid of $\bbN$ generated by $S$. 
It is a well-known fact that any additive subsemigroup of $\bbN$ is finitely 
generated and consequently, there exists a finite set of natural numbers $F_S$ such 
that $\langle S \rangle = \langle F_S \rangle$.

\begin{proposition}\label{prop:V_S} 
	The variety $\bV_S$ satisfies the inequality $x \leq x^{m+1}$ if and only if $m$ 
	belongs to $\langle S \rangle$. 
\end{proposition}

\begin{proof} Let $T$ be the set of all natural numbers $n$ such that $\bV_S$
satisfies the inequality $x \leq x^{n+1}$. First observe that $T$ is an additive
submonoid of $\bbN$. Indeed, if $\bV_S$ satisfies the inequalities $x \leq xx^m$ and
$x \leq xx^n$, then it satisfies $x \leq xx^m \leq (xx^n)x^m = x^{n+m+1}$. Now $T$
contains $S$ by definition and thus also $\langle S \rangle$. It follows that if $m$
belongs to $\langle S \rangle$, then $\bV_S$ satisfies the inequality $x \leq
x^{m+1}$.
	
Suppose now that $\bV_S$ satisfies the inequality $x \leq x^{m+1}$ and let 
\[
	L_S = \{a^{n+1} \mid n \in \langle S \rangle\}.
\]
Since $\langle S \rangle = \langle F_S \rangle$, one has
\[
  L_S = a\{a^s \mid s \in F_S\}^*
\]
and thus $L_S$ is a regular language. 

We claim that the ordered syntactic monoid $M$ of $L_S$ satisfies an inequality of
the form $x \leq x^{n+1}$ if and only if $n \in \langle S \rangle$. Suppose first
that $M$ satisfies $x \leq x^{n+1}$. Then the property $a \in L_S$ implies $a^{n+1}
\in L_S$ and hence $n \in \langle S \rangle$.

Conversely, let $n \in \langle S \rangle$. We need to prove that $M$ satisfies the
inequality $x \leq x^{n+1}$, or equivalently, that $a^k \leq_{L_S} (a^k)^{n+1}$ for
all $k \geq 0$. But for each $r \geq 0$, the condition $a^ra^k \in L_S$ implies $r +
k - 1 \in \langle S \rangle$. Since $r + k(n+1) - 1 = r + k - 1 + kn$, one gets $r +
k(n+1) - 1 \in \langle S \rangle$ and hence $a^r(a^k)^{n+1} \in L_S$ as required.
This concludes the proof of the claim.

In particular, since $M$ satisfies all the inequalities $x \leq x^{n+1}$ for $n \in
S$, $M$ belongs to $\bV_S$ and thus also satisfies the inequality $x \leq x^{m+1}$,
which finally implies that $m$ belongs to $\langle S \rangle$.
\end{proof}

\begin{corollary}\label{cor:V_S}
	Let $S$ and $T$ be two sets of natural numbers. Then $V_S = V_T$ if and
	only if $\langle S \rangle = \langle T \rangle$. 
\end{corollary}

\noindent It would also be interesting to have a systematic approach to treat
examples similar to those given in Section \ref{sec:Three examples}. That is, find an
algorithm which takes as input a monogenic ordered monoid $M$ and outputs a set of
inequalities defining respectively $\bV$, $\PdV$ and $\PdzV$, where $\bV$ is the
variety of ordered monoids generated by $M$.

\section*{Acknowledgements}

We would like to thank the anonymous referees for their helpful comments.

\bibliographystyle{amsplain}
\bibliography{Esik}

\end{document}